\documentclass[sigconf]{acmart}
\usepackage{mathtools}
\usepackage{multicol, multirow}
\usepackage{threeparttable}
\usepackage{booktabs}
\usepackage{amsmath}
\usepackage{amsfonts}
\DeclareMathOperator*{\argmax}{arg\,max}

\usepackage{ulem}
\usepackage{caption}
\usepackage{bbding}
\usepackage{enumitem}
\usepackage{graphics}
\usepackage{subfigure}
\usepackage{algorithm}
\usepackage{algpseudocode}
\usepackage{bm}

\newtheorem{prop}{Proposition}



\begin{document}
\title{Towards Automated Negative Sampling in Implicit Recommendation}

\author{Fuyuan Lyu}
\authornote{This work was done when the author worked as an intern at Huawei Noah's Ark Lab, Canada.}
\affiliation{
  \institution{McGill University}
  \city{Montreal}
  \country{Canada}
}
\email{fuyuan.lyu@mail.mcgill.ca}

\author{Yaochen Hu}
\affiliation{
  \institution{Huawei Noah's Ark Lab}
  \city{Montreal}
  \country{Canada}
}
\email{yaochen.hu@huawei.com}

\author{Xing Tang}
\affiliation{
  \institution{Huawei Noah's Ark Lab}
  \city{Shenzhen}
  \country{China}
}
\email{xing.tang@huawei.com}

\author{Yingxue Zhang}
\affiliation{
  \institution{Huawei Noah's Ark Lab}
  \city{Montreal}
  \country{Canada}
}
\email{yingxue.zhang@huawei.com}

\author{Ruiming Tang}
\affiliation{
  \institution{Huawei Noah's Ark Lab}
  \city{Shenzhen}
  \country{China}
}
\email{tangruiming@huawei.com}

\author{Xue Liu}
\affiliation{
  \institution{McGill University}
  \city{Montreal}
  \country{Canada}
}
\email{xueliu@cs.mcgill.ca}

\renewcommand{\shortauthors}{Fuyuan Lyu et al.}

\begin{abstract}
Negative sampling methods are vital in implicit recommendation models as they allow us to obtain negative instances from massive unlabeled data. 
Most existing approaches focus on sampling hard negative samples in increasingly complex ways. These studies tend to be orthogonal towards the recommendation model and implicit datasets. 
However, such an idea contradicts the common belief in automated machine learning that the model and dataset should be matched to achieve better performance.
Empirical experiments suggest that the best-performing negative sampler depends on the implicit dataset and the specific recommendation model. 
Hence, we propose a hypothesis that the negative sampler needs to be aligned with the capacity of the recommendation models as well as the statistics of the datasets to achieve optimal performance. A mismatch between these three would likely results in sub-optimal outcomes. 
An intuitive idea to address the mismatch problem is to exhaustively or manually search for the best-performing negative sampler given the model and dataset. However, such an approach is computationally expensive and time-consuming. Consequently, the efficient search for the optimal negative sampler remains an unsolved problem.
In this work, we propose the AutoSample framework that can automatically and adaptively select the best-performing negative sampler among a set of candidates. Specifically, we propose a loss-to-instance approximation to transform the negative sampler search task into the learning task over a weighted sum, enabling end-to-end training of the model. We also design an adaptive search algorithm to extensively and efficiently explore the search space, which is hard to model directly. A specific initialization approach is also obtained to better utilize the obtained model parameters during the search stage, which is similar to curriculum learning and leads to better performance and less computation resource consumption.
We evaluate the proposed framework on four real-world benchmark datasets over three widely-adopted basic models. Extensive experiments demonstrate the effectiveness and efficiency of our proposed framework. Further ablation studies are also conducted to deepen the understanding of our AutoSample framework from various perspectives.
\end{abstract}

\maketitle
\section{Introduction}

Collaborative Filtering, as the core component for personalized recommender systems, focuses on learning user preference from observed user-item interactions~\cite{BPR,AOBPR}. Nowadays, recommender systems heavily rely on implicit user feedback, such as purchases on E-commerce sites and views on content platforms. Implicit feedback is more accessible than explicit ones (such as rating) yet more challenging to utilize. Most datasets from implicit feedback only contain positive interactions (e.g., clicking an advertisement, finishing watching a video, etc. ), while negative feedback is necessary for training implicit recommendation~\cite{BPR}. Therefore various techniques are widely adopted to sample negative instances from the unobserved user item interactions~\cite{BPR,AOBPR,DNS,PNS,MixGCF}. An example of such an implicit recommender system is visualized in Figure~\ref{fig:normal}, where the recommendation model takes both the positive instances from the implicit dataset and negative instances from the negative sampler.

\begin{figure}[!htbp]
    \centering
    \includegraphics[width=0.42\textwidth]{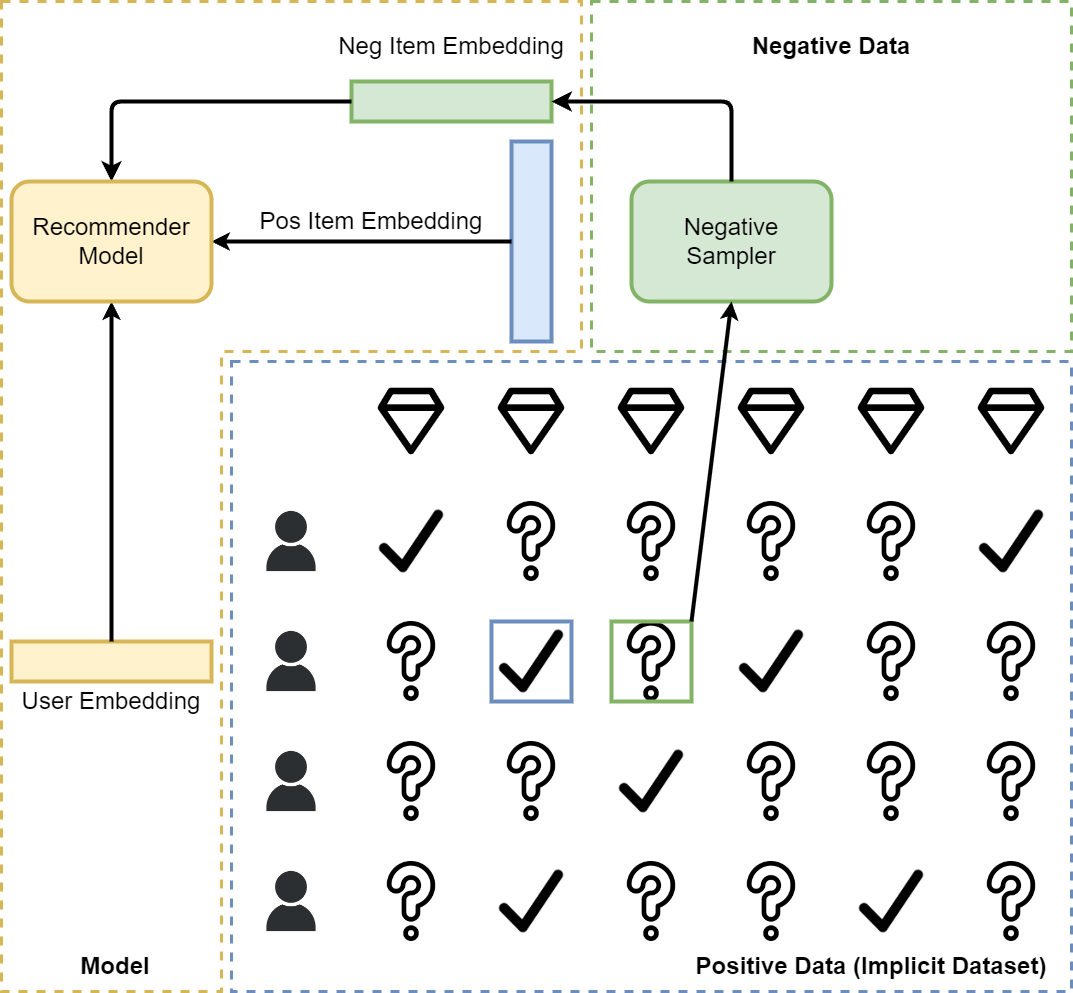}
    \caption{Illustration of An Implicit Recommender System.}
    \label{fig:normal}
\end{figure}

Existing negative sampling methods often replace the random sampling~\cite{BPR} with selecting hard negative instances from the pool of unobserved ones~\cite{AOBPR,DNS,SRNS}. Hard negative instances refer to instances that are difficult to distinguish for the recommendation model. Researchers employ various methods as the sampler for choosing hard negative instances, including but not limited to the recommendation model itself~\cite{DNS,AOBPR,SRNS}, the generative adversarial network~\cite{AdvIR,IRGAN} or graph model~\cite{MixGCF}. All these methods aim to identify instances that challenge the recommendation model's discriminatory ability.

However, these approaches often disregard the connection between the recommendation model, implicit datasets, and negative sampling methods. We challenge this idea as it contradicts the common belief in standard machine learning tasks, such as computer vision~\cite{DARTS} and natural language processing~\cite{nas-nlp}, where the model architecture should be typically selected to fit the target dataset to achieve good performance. Unlike these tasks, the implicit recommendation only contains the positive instances of the training dataset, as depicted in Figure \ref{fig:normal}, while the negative sampler generates the negative instances. Since the positive instances of the dataset are fixed, we hypothesize that the negative sampler should be aligned with the capacity of the recommendation models as well as the statistics of the dataset's positive part to achieve optimal performance. A mismatch among these three components would likely results in sub-optimal outcomes. This hypothesis can be empirically validated by the experimental result in Table \ref{tab:intuition}, where the best-performing negative sampler tends to depend on the dataset and model.

\begin{table}[!htbp]
    \centering
    \caption{The best-performing negative samplers in Table \ref{tab:main}.}
    \begin{tabular}{c|c|c|c|c}
    \hline
        Model/Dataset               & TB2014 & TB2015 & Alibaba & Amazon \\
    \hline
        MF~\cite{BPR}               & RNS   & RNS       & PNS       & PNS       \\
        LightGCN~\cite{LightGCN}    & RNS   & MixGCF    & RNS       & RNS       \\
        NGCF~\cite{NGCF}            & RNS   & MixGCF    & MixGCF    & MixGCF    \\
    \hline
    \end{tabular}
    \begin{tablenotes}
    \footnotesize
    \item [1] Here RNS, PNS and MixGCF indicates random sampler~\cite{BPR}, popularity-based sampler~\cite{PNS} and MixGCF~\cite{MixGCF}.
    \end{tablenotes}
    \label{tab:intuition}
\end{table}

An intuitive approach to selecting the suitable negative sampler given the recommendation model and target dataset is to perform an exhaustive and manual search for the best-performing one. However, such a solution would require significant computation resources and time. To mitigate the search cost, we propose an automated negative sampling problem that formulates the selection of negative samplers as a neural architecture search problem, as shown in Figure \ref{fig:auto}. 
Directly solving the automated negative sampling problem poses significant challenges, especially in maintaining the end-to-end training flow throughout the search process. To address this challenge and enable end-to-end training, we propose the \textit{instance-to-loss} transformation, which reformulates the problem as a weight-sum of multiple losses generated by different samplers.

\begin{figure}[!htbp]
    \centering
    \includegraphics[width=0.42\textwidth]{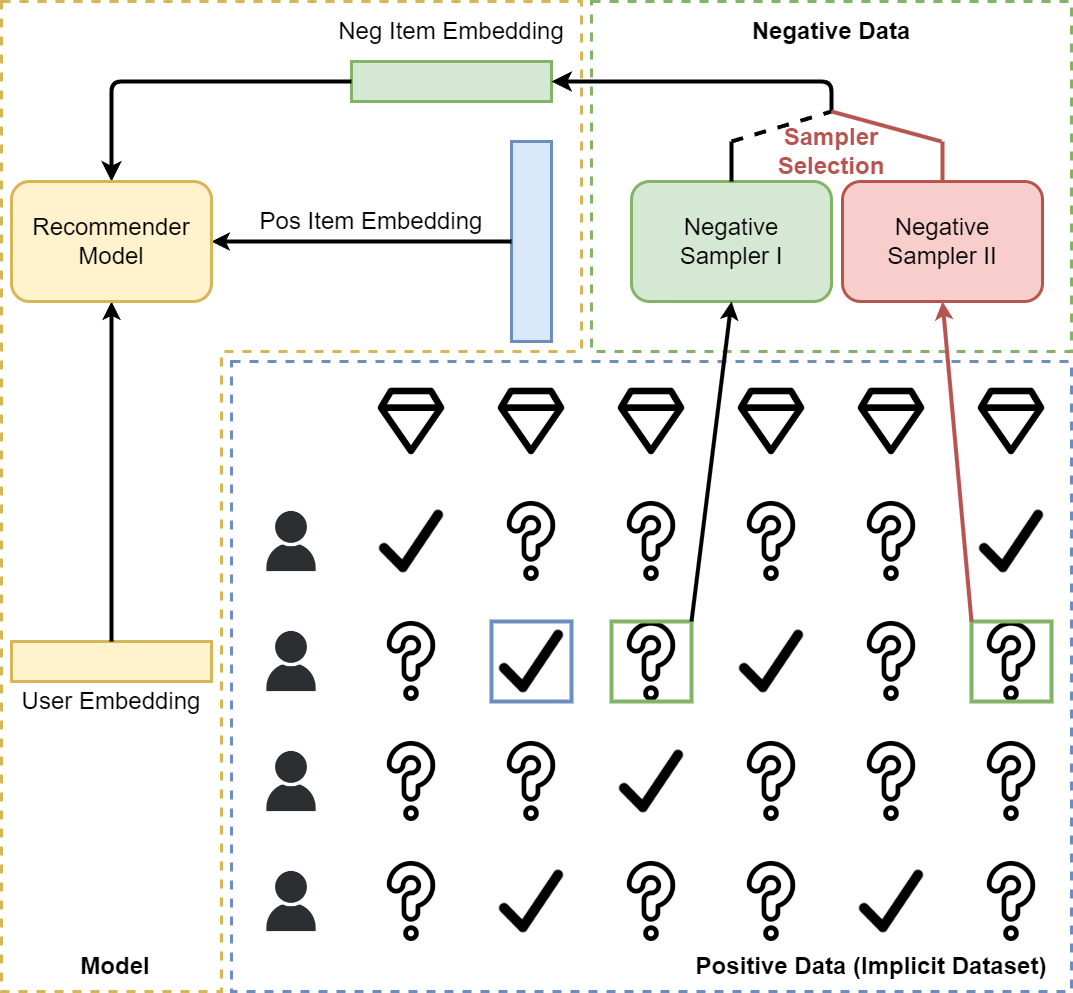}
    \caption{Illustration of Automated Negative Sampling Problem. We highlight the difference with Figure \ref{fig:normal} in Red.}
    \label{fig:auto}
\end{figure}

To tackle the automated negative sampling problem, we introduce a method named \textbf{AutoSample}, which automatically and adaptively selects the sampling strategies for different models and datasets. A gradient-based search algorithm is proposed to explore the search space efficiently. Additionally, inspired by curriculum learning, we propose a specialized retraining scheme to facilitate the model parameters during the search stage, leading to improved performance.
We summarize our major contributions as follows:
\begin{itemize}[topsep=0pt,noitemsep,nolistsep,leftmargin=*]
    \item This paper first proposes the hypothesis that the performance of negative samplers is dependent on the specific model and dataset, which challenges the existing research approach of developing the negative sampler without considering the model's capability or dataset's statistic. 
    \item This paper introduces the automated negative sampling problem, which can be viewed as a framework that can be combined with existing negative samplers. To ensure end-to-end trainability, we propose the instance-to-loss approximation. Additionally, we present a gradient-based search algorithm and a specialized retraining scheme to efficiently and effectively explore the search space.
    \item The extensive experiments are conducted on four public datasets and three basic models. The empirical results demonstrate the feasibility and effectiveness of our proposed framework.
\end{itemize}

\section{Methodology}
\label{sec:method}

In this section, we first formulate the implicit recommendation problem in Section~\ref{sec:method_formulation}. Then, we formulate the automated negative sampling problem in Section~\ref{sec:method_problem}. Finally, we introduce our method AutoSample in Section~\ref{sec:method_autosample}.

\subsection{Implicit Recommendation Formulation}
\label{sec:method_formulation}

Suppose we have a user set $\mathbf{U}$, including $M$ users, and an item set $\mathbf{I}$, including $N$ items. 
In the implicit recommendation settings, we only have the positive interaction set $\mathbf{D}^p \subseteq \mathbf{U}\times\mathbf{I}$ with each pair $(u, i)\in\mathbf{D}^p$ indicating user $u$ has interacted with item $i$. 
To facilitate the training, a negative interaction set $\mathbf{D}^n \subseteq \mathbf{D}_{-}$ needs to be sampled from the candidate space $\mathbf{D}_{-} \triangleq \mathbf{U} \times \mathbf{I} \ - \ \mathbf{D}^p$.
An implicit recommendation model generally consists of three components: 
(i) a scoring function $\mathcal{S}(u, i)$ to model the relevance between user $u$ and item $i$, 
(ii) a loss function $\mathcal{L}$ over the positive pair $(u, i)\in\mathbf{D}^p$ and negative pair $(u, j)\in\mathbf{D}^n$, and 
(iii) negative sampler $\pi(u, k, \mathbf{D}_{-})$ to generate the negative interaction set $\mathbf{D}^n_u = \{(u, j)| j\in \mathbf{I} - \mathbf{I}_u\}$ for each user $u$ with cardinality $|\mathbf{D}^n_u| = k$, where $\mathbf{I}_u=\{i |(u,i)\in\mathbf{D}^p\}$ is the set of iteracted items for user $u$. 
Eventually $\mathbf{D}^n\triangleq\cup_u\mathbf{D}^n_u$. Since we focus on the same candidate space $\mathbf{D}_-$, we omit $\mathbf{D_-}$ for simplicity as $\pi(u, k)$. 

\subsubsection{Scoring Function}
The scoring function $\mathcal{S}(u,i|W)$ aims to calculate the relevance score between user $u$ and item $i$, parameterized by $W$. There exist various scoring functions based on matrix factorization~\cite{BPR,AOBPR}, multiple-layer perceptron~\cite{NCF} and graph neural network~\cite{LightGCN,NGCF}. 

\subsubsection{Loss Function}
For the loss function $\mathcal{L}$, Bayesian Personalized Loss(BPR)~\cite{BPR} is widely adopted to measure the difference between positive and negative items given the sample user $u$. It can be formulated as follows:
\begin{equation}
\label{eq:loss_bpr}
\begin{aligned}
    \mathcal{L}(\mathbf{D}^p, \mathbf{D}^n, W) = -\sum_{(u,i) \in \mathbf{D}^p} \sum_{(u,j) \in \mathbf{D}^n_u} \ln \sigma \left(\mathcal{S}(u,i|W) - \mathcal{S}(u,j|W)\right),
\end{aligned}
\end{equation}
where $\sigma(x) = \frac{1}{1+ e^{-x}}$ indicates the sigmoid function. The overall idea for the above equation is to maximize the distance between positive and negative items given the user and scoring function.

\subsubsection{Negative Sampler}
The negative sampler $\pi(u,k)$ aims to select $k$ negative samples $\mathbf{D}_u^n$ for user $u$. 
We consider samplers that draw $k$ samples with the same distribution independently, so we only focus on the sampling process for one negative pair $(u,j)\in \mathbf{U} \times \mathbf{I} - \mathbf{D}_p$. 
To select one negative sample pair $(u, j)$, the user $u$ is first determined. Here we usually take $u$ for one positive pair $(u, i)$ to fit the BPR loss~\cite{BPR}. The selection of the negative item $j$, on the other hand, is based on the negative sampling distribution $\mathbf{p}_u$ over each candidate negative samples $j \in \mathbf{I} - \mathbf{I}_u$ for user $u$, where 
\begin{equation}
\label{eq:bpr_loss}
    \mathbf{p}_u(j) > 0, \ \sum_{j \in \mathbf{I} - \mathbf{I}_u} \mathbf{p}_u(j) = 1.
\end{equation}

Various sampling probability has been proposed to facilitate the training of the implicit recommendation model. The most commonly used probability is the uniform distribution~\cite{BPR}, where all candidate negative samples have an equal chance to be selected. However, this method usually suffers from low-quality samples~\cite{DNS,SRNS,MixGCF}. To solve this, various methods have been proposed. PNS~\cite{PNS} heuristically selects the negative samples based on the items' popularity. Hard negative sampling methods, like DNS~\cite{DNS} or AOBPR~\cite{AOBPR} explicitly favour the candidate samples with higher $\mathbf{S}(u,j)$, which contains more information. There also exist works like IRGAN~\cite{IRGAN} and AvdIR~\cite{AdvIR}, which simultaneously train a parameterized function $\mathbf{p}_u$ to minimize the loss in Equation ~\ref{eq:bpr_loss}, based on GAN. Different choices of $\mathbf{p}_u$ are listed in Table \ref{table:sampler}. 

Hence, the final training objective for implicit recommendation can be summarized as follows:
\begin{equation}
\label{eq:loss_rec}
    \min_W \ \mathcal{L}(\mathbf{D}^p, \mathbf{D}^n, W), 
\end{equation}
where $\mathbf{D}^n=\cup_u \mathbf{D}_u^n$ and $\mathbf{D}_u^n$ is generated by the negative sampler $\pi(u, k)$.




\subsection{Automated Negative Sampling}
\label{sec:method_problem}

In this section, we introduce the proposed automated negative sampling problem. The intuition is to adaptively and automatically find the most suitable negative sampler. Consider $T$ different candidate samplers $\pi_1(u, k), \pi_2(u,k),\ldots, \pi_T(u,k)$. To automatically search for the optimal sampler, we assign a learnable weight $\alpha_t\ge 0$ for each sampler $\pi_t(u,k)$. Derived from Equation \ref{eq:loss_rec}, the automated negative sampling can be formulated as:


\begin{align}
    \min_{W, \bm{\alpha}} \quad&\mathcal{L}(\mathbf{D}^p, \mathbf{D}^n, W), \label{eq:loss_sample}\\
    \text{s.t.} \quad& \mathbf{D}^n=\cup_u \mathbf{D}_u^n, \quad \mathbf{D}_u^n = \cup_t \mathbf{D}_{u,t}^n,\quad \forall u,\nonumber\\
    & \mathbf{D}_{u,t}^n \text{ is generated by } \pi_t(u, \alpha_t k), \quad \forall u, t\label{eq:sampling}\\
    & \sum_{t=1}^T \alpha_t = 1 \quad \text{and} \quad \alpha_t \ge 0, \ \forall t.\nonumber
\end{align}

Here $\mathbf{D}_{u,t}^n$ denotes the negative samples generated by sampler $\pi_t(u, \alpha_t k)$ for user $u$. The search parameter $\bm{\alpha} = [\alpha_1, \alpha_2, \cdots, \alpha_T]$ measures the candidates' contribution to the final negative interaction set $\mathbf{D}^n$. To determine $\bm{\alpha}$, the naive solution with an exhaustive or manual search of the continuous weight from a large search space is extremely costly. Besides, this formulation is intractable to simultaneously train  $W$ and $\bm{\alpha}$ in an end-to-end manner since the sampling process in Equation~\ref{eq:sampling} cut off the back-propagation path for parameter $\bm{\alpha}$, losing the useful gradient information.

\subsubsection{Instance-level to Loss-level Transformation}
To make the recommendation model end-to-end trainable with $\bm{\alpha}$, we propose a method to transform the loss function to take $\bm{\alpha}$ from the samplers (instance-level) to the loss (loss-level). Let $\hat{\mathbf{D}}_{u,t}^n$ denotes the negative samples generated by $\pi_t(u, k)$ for $\forall u, t$, and denote $\hat{\mathbf{D}}_{t}^n = \cup_u\hat{\mathbf{D}}_{u,t}^n$. Then $\mathcal{L}(\mathbf{D}^p, \mathbf{\hat{D}}^n_t, W)$ is the loss with only sampler $\pi_t(u, k)$. We have

\begin{prop} \label{prop:transform}
Given the set of mutually independent samplers $\{\pi_t(u, k) | t = 1, 2, \ldots, T\}$, we have
    \begin{align}
        & \mathbb{E}_{\{\pi_t(u, \alpha_tk)|\forall u,t\}}\mathcal{L}(\mathbf{D}^p, \mathbf{D}^n, W)\nonumber\\
        = &\sum_t \alpha_t \mathbb{E}_{\{\pi_t(u, k)|\forall u\}} \mathcal{L}(\mathbf{D}^p, \mathbf{\hat{D}}^n_t, W)
    \end{align}
\end{prop}
\begin{proof}
Denote $\delta_{uij} \triangleq \ln \sigma(\mathcal{S}(u,i|W) - \mathcal{S}(u,j|W))$. We have
\begin{align}
 & \mathbb{E}_{\{\pi_t(u, \alpha_tk)|\forall u,t\}}\mathcal{L}(\mathbf{D}^p, \mathbf{D}^n, W)\nonumber\\
 = &\mathbb{E}_{\{\pi_t(u, \alpha_tk)|\forall u,t\}}\left(-\sum_{(u,i) \in \mathbf{D}^p} \sum_{(u,j) \in \mathbf{D}^n_u} \delta_{uij}\right) \quad \text{(by (\ref{eq:loss_bpr}))}\nonumber\\
 = &\mathbb{E}_{\{\pi_t(u, \alpha_tk)|\forall u,t\}}\left(-\sum_{(u,i) \in \mathbf{D}^p}\sum_t \sum_{(u,j) \in \mathbf{D}^n_{u,t}} \delta_{uij}\right)\nonumber\\
 =& \sum_t \mathbb{E}_{\{\pi_t(u, \alpha_tk)|\forall u\}}\left(-\sum_{(u,i) \in \mathbf{D}^p} \sum_{(u,j) \in \mathbf{D}^n_{u,t}} \delta_{uij}\right)\label{eq:shift_proof_1}\\
 =& \sum_t \alpha_t\mathbb{E}_{\{\pi_t(u, k)|\forall u\}}\left(-\sum_{(u,i) \in \mathbf{D}^p} \sum_{(u,j) \in \hat{\mathbf{D}}^n_{u,t}} \delta_{uij}\right)\label{eq:shift_proof_2}\\
 =& \sum_t \alpha_t \mathbb{E}_{\{\pi_t(u, k)|\forall u\}} \mathcal{L}(\mathbf{D}^p, \mathbf{\hat{D}}^n_t, W). \nonumber
\end{align}
Equtation \ref{eq:shift_proof_1} comes from the linearity of expectation and the independence of the samplers. Equtation \ref{eq:shift_proof_2} comes from the fact that $\pi_t(u, k)$ draws $k$ samples independantly for $\forall u, t$.
\end{proof}

Given Proposition~\ref{prop:transform}, we can safely rewrite Problem~\ref{eq:loss_sample} as
\begin{align}
    \min_{W, \bm{\alpha}} \quad&\sum_t \alpha_t \mathcal{L}(\mathbf{D}^p, \mathbf{\hat{D}}^n_t, W), \label{eq:loss_loss}\\
    \text{s.t.} \quad& \mathbf{\hat{D}}^n_t=\cup_u \mathbf{\hat{D}}_{u,t}^n, \forall u,\nonumber\\
    & \mathbf{\hat{D}}_{u,t}^n \text{ is generated by } \pi_t(u, k), \quad \forall u, t \nonumber\\
    & \sum_{t=1}^T \alpha_t = 1 \quad \text{and} \quad \alpha_t \ge 0, \ \forall t.\nonumber
\end{align}
where $\bm{\alpha}$ can be trained in an end-to-end manner. Proposition~\ref{prop:transform} guarantees that this substitution is a decent replacement for the original loss.

\begin{table}[htbp]
    \caption{Comparison of Different Negative Samplers.}
    \centering
    \begin{tabular}{cccc}
    \hline
        $\pi$ & $\mathbf{p}_u(j)$ & Category \\
    \hline
        RNS~\cite{BPR}      & $= |\mathbf{I} - \mathbf{I}_u|^{-1}$  & Heuristic      \\
        PNS~\cite{PNS}      & $\propto (\text{pop}_j)^\beta$       & Heuristic      \\
        AOBPR~\cite{AOBPR}  & $\propto \exp(-\text{grk}(u,j)/\lambda)$ & Hard Nega  \\
        DNS~\cite{DNS}      & $\propto \exp(-\text{lrk}(u,j)/\lambda)$ & Hard Nega  \\
        IRGAN~\cite{IRGAN}  & learnable     & Adversarial Hard Nega  \\   
        AdvIR~\cite{AdvIR}  & learnable     & Adversarial Hard Nega  \\
        SRNS~\cite{SRNS}    & learnable     & Hard Nega \\
        MixGCF~\cite{MixGCF}& learnable     & Graph-based Hard Nega \\
    \hline
    \end{tabular}
    \begin{tablenotes}
    \footnotesize
    \item [1] Here \textit{grk()} and \textit{lrk()} refer to global and local ranking, respectively. \textit{pop} stands for popularity.
    \end{tablenotes}
    \label{table:sampler}
\end{table}

\subsection{AutoSample}
\label{sec:method_autosample}
In this section, we propose AutoSample, which addresses the automated negative sampling problem discussed in the previous section. Following the convention from the classical neural architecture search problem~\cite{DARTS}, we introduce the search space, the search algorithm, the search process and the re-training process, respectively.

\subsubsection{Search Space}
\label{sec:method_search_space}
The search space determines which sampler needs to be taken as a potential candidate. In this experiment, we select the candidate based on two criteria: (i) sampling performance and (ii) sampling efficiency. The sampling performance is usually measured by the metrics given the corresponding sampler, whereas the sampling efficiency is important as it may affect the total training time, especially when the dataset is large. The general selection criteria for candidate negative samplers are that we select the top-T best-performing negative samplers as candidates. We discuss the influence of the search space empirically in Section \ref{sec:search_space}. Notes that, to improve model efficiency, we only select negative samplers from classic negative samplers with relatively high efficiency, specifically RNS~\cite{BPR}, PNS~\cite{PNS}, DNS~\cite{DNS} and AOBPR~\cite{AOBPR}.

\subsubsection{Search Algorithm}
\label{sec:method_search_algorithm}
In this section, we illustrate the search algorithm for efficiently exploring the search space. Instead of exhaustively searching for the optimal search parameter, which makes the whole framework not end-to-end differentiable, we approximate the search of negative sampler via the Gumbel-softmax operation~\cite{Gumbel-Softmax} in this work. The Gumbel-softmax operation provides a differentiable sampling, which makes the parameters $\bm{\alpha}$ learnable.

To be specific, suppose architecture parameters $\{\alpha_{t} |t =1, 2,\ldots, T\}$ are the search probability over different negative samplers, $\Pi\triangleq \{\pi_t|\forall t\}$ indicates the search space over all possible negative sampling methods. Then a discrete selection $z$ can be drawn via the Gumbel-softmax trick\cite{Gumbel-Softmax-dist} as

\begin{equation}
\begin{aligned}
    & z = \text{onehot} \left(\argmax_{\pi_t \in \Pi} [ \log \alpha_{t} + g_{t} ] \right), \\
    & g_{t} = -\log(-\log(u_{t})), \\
    & u_{t} \sim \text{Uniform}(0,1).
\end{aligned}
\end{equation}

The independent and identically distributed (i.i.d) \textit{gumbel noises} $g_t$ disturb the {$\log \alpha_t$} terms. Also, they make the $\argmax$ operation equivalent to drawing a sample from the search weight $\alpha_1, \alpha_2, \cdots \alpha_T$. 
However, because of the $\argmax$ operation, this sampling method is non-differentiable. 
We tackle this problem by straight-through Gumbelsoftmax~\cite{Gumbel-Softmax-dist}, which leverages a softmax function as a differentiable approximation to the $\argmax$ operation:
\begin{equation}
\label{eq:probability}
    p_t = \frac{exp((\log (\alpha_t) + g_t) / \tau)}{\sum_{s=1}^T exp((\log (\alpha_s) + g_s) / \tau)}, \quad \forall t,
\end{equation}

where $p_t$ denotes the probability of selecting $t$-th negative sampler $\pi_t$. The temperature $\tau$ is the temperature parameter to control the smoothness of the Gumbel-softmax operation. When $\tau$ approximates to zero, the Gumbel-softmax operation approximately outputs a one-hot vector. Then we can re-write the loss term in Equation \ref{eq:loss_loss} as
\begin{equation}
\label{eq:loss_final}
    \mathcal{L}(\mathbf{D}^p, \mathbf{D}^n, W) = \sum_t \alpha_t \mathcal{L}(\mathbf{D}^p, \mathbf{\hat{D}}^n_t, W) \approx \sum_{t=1} p_t \mathcal{L}(\mathbf{D}^p, \mathbf{\hat{D}}_t^n, W),
\end{equation}
where $\hat{\mathbf{D}}_{t}^n = \cup_u\hat{\mathbf{D}}_{u,t}^n$ and $\hat{\mathbf{D}}_{u,t}^n$ is the negative samples generated by $\pi_t(u, k)$. In conclusion, the search algorithm becomes end-to-end differentiable after introducing the Gumbel-softmax operation.


\subsubsection{Search Process Optimization}
\label{sec:method_search_process}

In the above subsections, we formulate and transform the instance-level negative sampler search problem to loss-level. Then we transform the loss-level formulation into a neural architecture search problem and introduce the Gumbel-Softmax trick to make the recommendation model end-to-end differentiable. Now we discuss how we optimize the model during the search process. 

In AutoSample, two sets of parameters need to be optimized: the model parameter $W$ contained in scoring function $\mathcal{S}(\cdot, \cdot)$ and the search parameter $\bm{\alpha}$. Note that $p_t$ is directly generated by the Gumbel-Softmax operation based on Equation \ref{eq:probability}. In the naive DARTS~\cite{DARTS} setting, two parameters are iteratively trained over different batches. However, in the implicit recommendation, the dataset only contains positive interactions, and it is hard to directly apply such a method. Inspired by the previous work~\cite{AutoFIS,OptInter}, we jointly train both parameters over the same training batch. Such a process can be formulated in Algorithm~\ref{alg:search}.
\begin{algorithm}
    \caption{The Optimization of Search Process}
    \label{alg:search}
    \begin{algorithmic}[1]
    	\Require Implicit dataset $\mathbf{D}^p$
    	\Ensure the well-trained search parameter $\bm{\alpha}^{*}$ and model parameter $W^{\prime}$
        \While {not converged}
            \State Sample a mini-batch of training data $\mathcal{B}^{p}$ from the implicit
            \Statex \qquad dataset $\mathbf{D}^p$
            \State Generate negative mini-batch $\mathcal{B}^{n}$ given the current 
            \Statex \qquad search parameter $\bm{\alpha}$ and candidate samplers $\Pi$
            \State Calculate the loss $\mathcal{L}(\mathcal{B}^p, \mathcal{B}^n, W)$ given Equation \ref{eq:loss_loss}
            \State Update model parameter $W$ by descending gradient 
            \Statex \qquad $\nabla_{W} \mathcal{L}(\mathcal{B}^p, \mathcal{B}^n, W)$
            \State Update search parameter $\bm{\alpha}$ by descending gradient
            \Statex \qquad $\nabla_{\bm{\alpha}} \mathcal{L}(\mathcal{B}^p, \mathcal{B}^n, W)$
        \EndWhile
    \end{algorithmic}
\end{algorithm}

\subsubsection{Retraining Process}
\label{sec:method_retrain_process}

After obtaining the optimal probability for each sampler $\bm{\alpha}^*$, we retrain the model with the best-performing parameter $W^{'}$ as initialization and freeze $\bm{\alpha}^*$ as the search result for the negative sampler. The idea for keeping the model parameter instead of dropping them as typical neural architecture search methods would is that as $\bm{\alpha}$ gradually converge to the optimal, the model parameter $W$ is also adaptively trained from easy negative instances to hard ones, similar to curriculum learning~\cite{CL-Survey}. It has also been proven that through careful selection of the initialization parameter, a deep recommender system can achieve better performance than random initialized~\cite{OptFS}.
The re-training process can be summarized as Algorithm \ref{alg:retrain}. 

\begin{algorithm}
\caption{The Re-training Process}
    \label{alg:retrain}
	\begin{algorithmic}[1]
        \Require Implicit dataset $\mathbf{D}^p$, optimal search parameter $\bm{\alpha}^*$, initialization model parameter $W^{\prime}$
        \Ensure the well-trained model parameter $W^*$
        \State select the optimal negative sampler $\pi^*$ given the optimal search parameter $\bm{\alpha}^*$ and candidate samplers $\Pi$
        \State Re-initialize the model with parameter $W^{\prime}$
        \While {not converged}
            \State Sample a mini-batch of training data $\mathcal{B}^{p}$ from the implicit
            \Statex \qquad dataset $\mathbf{D}^p$
            \State Generate negative mini-batch $\mathcal{B}^{n}$ given sampler $\pi^*$
            \State Calculate the loss $\mathcal{L}(\mathcal{B}^p, \mathcal{B}^n, W)$ given Equation \ref{eq:loss_rec}
            \State Update model parameter $W$ by descending gradient 
            \Statex \qquad $\nabla_{W} \mathcal{L}(\mathcal{B}^p, \mathcal{B}^n, W)$
        \EndWhile
	\end{algorithmic}
\end{algorithm}

\section{Experiment}

In this section, to comprehensively evaluate our proposed framework, we design experiments to answer the following research questions: 
\begin{itemize}[topsep=0pt,noitemsep,nolistsep,leftmargin=*]
    \item \textbf{RQ1}: Can the AutoSampler outperform all its candidate samplers and other SOTA baselines?
    \item \textbf{RQ2}: Does the increasing amount of negative instances cause the improvement of AutoSampler? 
    \item \textbf{RQ3}: Is our selection algorithm more efficient than exhaustive search? 
    \item \textbf{RQ4}: How effective is our two-stage learning algorithm?
    \item \textbf{RQ5}: What kind of samplers does our method perform?
\end{itemize}

\subsection{Experimental Details}

\subsubsection{Dataset Description}

In this paper, we evaluate our method on the following benchmark datasets: Taobao2014\footnote{https://tianchi.aliyun.com/dataset/46}, Taobao2015\footnote{https://tianchi.aliyun.com/dataset/53}, Alibaba~\cite{MCNS} and Amazon\footnote{http://jmcauley.ucsd.edu/data/amazon/links.html}. We randomly split the datasets into training, validation and testing sets by 3:1:1. The statistics and descriptions of these datasets are summarized in Appendix \ref{ref:appendix_exp_dataset}. 

\subsubsection{Evaluation Metrics}
We choose four widely-adopted evaluation metrics: Recall@K(R@K), NDCG@K(N@K), Precision@K(P@K) and Hit Ratio@K(H@K). We choose K=20 as the default setting. 

\subsubsection{Recommendation Models} 
All experiments are performed on the three basic recommendation models: Matrix Factorization~\cite{BPR}, LightGCN~\cite{LightGCN} and NGCF~\cite{NGCF}. We further compare our methods with six negative sampling baselines: RNS~\cite{BPR}, PNS~\cite{PNS}, AOBPR~\cite{AOBPR}, DNS~\cite{DNS}, SRNS~\cite{SRNS} and MixGCF~\cite{MixGCF}. Notes that we exclude commonly-used adversarial hard negative methods like IRGAN~\cite{IRGAN} or AdvIR~\cite{AdvIR} as they are constantly outperformed by SRNS~\cite{SRNS} and MixGCF~\cite{MixGCF}. Details about the recommendation models and negative sampling baselines are shown in Appendix \ref{ref:appendix_exp_rec} and \ref{ref:appendix_exp_baseline}. We also present the reproducibility details in Appendix \ref{ref:appendix_exp_repro}.

\begin{table*}[!htbp]
    \centering
    \caption{Overall Performance Comparison}
    \resizebox{.97\textwidth}{!}{
    \begin{tabular}{c|c|cccc|c|cccc|c}
    \hline
         & \multirow{2}{*}{Sampler} & \multicolumn{5}{c|}{Taobao2014} & \multicolumn{5}{c}{Taobao2015} \\
    \cline{3-12}
         & & R@20 & N@20 & P@20 & H@20 & Rank & R@20 & N@20 & P@20 & H@20 & Rank  \\
    \hline
    \multirow{6}{*}{\rotatebox{90}{MF}}
        & RNS   & 0.0238(2) & 0.0247(1) & 0.0155(1) & 0.2403(1) & 1.25 & 0.0627(2) & 0.0412(2) & 0.0094(2) & 0.1701(2) & 2.25 \\
        & PNS   & 0.0220(4) & 0.0212(4) & 0.0128(4) & 0.2013(4) & 4    & 0.0421(6) & 0.0248(6) & 0.0065(6) & 0.1170(6) & 6 \\
        & DNS   & 0.0093(6) & 0.0093(6) & 0.0059(6) & 0.1044(6) & 6    & 0.0556(4) & 0.0406(4) & 0.0082(4) & 0.1516(4) & 4 \\
        & AOBPR & 0.0111(5) & 0.0115(5) & 0.0073(5) & 0.1261(5) & 5    & 0.0516(5) & 0.0375(5) & 0.0077(5) & 0.1417(5) & 5 \\
        & SRNS  & 0.0221(3) & 0.0221(3) & 0.0142(3) & 0.2207(3) & 3    & 0.0615(3) & 0.0417(2) & 0.0092(3) & 0.1669(3) & 2.75 \\
    \cline{2-12}
        & AutoSample & 0.0243$^*$(1) & 0.0242(2) & 0.0153(2) & 0.2355(2) & 1.75 & 0.0688$^*$(1) & 0.0449$^*$(1) & 0.0103$^*$(1) & 0.1854$^*$(1) & 1  \\
    \hline
    \multirow{7}{*}{\rotatebox{90}{LightGCN}}
        & RNS   & 0.0327(1) & 0.0328(1) & 0.0205(2) & 0.2959(2) & 1.5 & 0.0696(4) & 0.0479(4) & 0.0105(4) & 0.1873(4) & 4 \\
        & PNS   & 0.0267(5) & 0.0264(5) & 0.0170(5) & 0.2535(5) & 5   & 0.0652(5) & 0.0454(5) & 0.0097(5) & 0.1757(5) & 5 \\
        & DNS   & 0.0090(7) & 0.0091(7) & 0.0057(7) & 0.1018(7) & 6   & 0.0554(7) & 0.0405(7) & 0.0081(7) & 0.1510(7) & 7 \\
        & AOBPR & 0.0133(6) & 0.0144(6) & 0.0099(6) & 0.1628(6) & 6   & 0.0591(6) & 0.0426(6) & 0.0088(6) & 0.1602(6) & 6 \\
        & SRNS  & 0.0316(3) & 0.0307(3) & 0.0192(3) & 0.2847(3) & 3   & 0.0716(3) & 0.0489(3) & 0.0107(3) & 0.1917(3) & 3 \\
        & MixGCF& 0.0288(4) & 0.0297(4) & 0.0186(4) & 0.2755(4) & 4   & 0.0755(1) & 0.0504(1) & 0.0113(1) & 0.2010(1) & 1 \\
    \cline{2-12}
        & AutoSample & 0.0321(2) & 0.0323(2) & 0.0206(1) & 0.2982$^*$(1) & 1.5 & 0.0725(2) & 0.0493(2) & 0.0108(2) & 0.1932(2) & 2 \\ 
    \hline
    \multirow{7}{*}{\rotatebox{90}{NGCF}}
        & RNS   & 0.0270(3) & 0.0263(2) & 0.0173(2) & 0.2609(2) & 2 & 0.0572(4) & 0.0356(4) & 0.0087(3) & 0.1557(2) & 3 \\
        & PNS   & 0.0268(4) & 0.0245(4) & 0.0153(4) & 0.2378(4) & 4 & 0.0395(7) & 0.0225(6) & 0.0062(6) & 0.1118(5) & 5.75 \\
        & DNS   & 0.0088(7) & 0.0089(7) & 0.0056(7) & 0.1025(7) & 7 & 0.0556(5) & 0.0407(3) & 0.0082(4) & 0.1517(3) & 3.5 \\
        & AOBPR & 0.0100(6) & 0.0108(6) & 0.0076(6) & 0.1305(6) & 6 & 0.0452(6) & 0.0273(5) & 0.0068(5) & 0.1260(4) & 4.75 \\
        & SRNS  & 0.0227(5) & 0.0206(5) & 0.0129(5) & 0.2023(5) & 5 & 0.0594(3) & 0.0390(3) & 0.0089(3) & 0.1615(3) & 3 \\
        & MixGCF& 0.0274(2) & 0.0259(3) & 0.0166(3) & 0.2501(3) & 3 & 0.0614(2) & 0.0429(2) & 0.0092(2) & 0.1671(2) & 2 \\
    \cline{2-12}
        & AutoSample & 0.0308$^*$(1) & 0.0293$^*$(1) & 0.0187$^*$(1) & 0.2747$^*$(1) & 1 & 0.0706$^*$(1) & 0.0481$^*$(1) & 0.0106$^*$(1) & 0.1904$^*$(1) & 1 \\
    \hline
    \hline
         & \multirow{2}{*}{Sampler} & \multicolumn{5}{c|}{Alibaba} & \multicolumn{5}{c}{Amazon} \\
    \cline{3-12}
         & & R@20 & N@20 & P@20 & H@20 & Rank & R@20 & N@20 & P@20 & H@20 & Rank  \\
    \hline
    \multirow{6}{*}{\rotatebox{90}{MF}}
        & RNS   & 0.0398(3) & 0.0178(3) & 0.0024(2) & 0.0462(3) & 3     & 0.0310(3) & 0.0140(3) & 0.0018(2) & 0.0355(3) & 2.75  \\
        & PNS   & 0.0410(2) & 0.0187(2) & 0.0024(2) & 0.0475(2) & 2     & 0.0317(2) & 0.0147(2) & 0.0018(2) & 0.0358(2) & 2     \\
        & DNS   & 0.0201(4) & 0.0084(4) & 0.0013(4) & 0.0253(4) & 4     & 0.0217(4) & 0.0073(5) & 0.0012(4) & 0.0244(4) & 4.25  \\
        & AOBPR & 0.0115(6) & 0.0047(6) & 0.0007(5) & 0.0144(5) & 5.5   & 0.0161(6) & 0.0066(6) & 0.0009(6) & 0.0177(6) & 6     \\
        & SRNS  & 0.0124(5) & 0.0056(5) & 0.0007(5) & 0.0142(6) & 5.25  & 0.0198(5) & 0.0098(4) & 0.0011(5) & 0.0222(5) & 4.75  \\
    \cline{2-12}
        & AutoSample & 0.0465$^*$(1) & 0.0211$^*$(1) & 0.0028$^*$(1) & 0.0543$^*$(1) & 1 & 0.0354$^*$(1) & 0.0155$^*$(1) & 0.0020(1) & 0.0391$^*$(1) & 1  \\
    \hline
    \multirow{7}{*}{\rotatebox{90}{LightGCN}}
        & RNS   & 0.0630(2) & 0.0284(3) & 0.0038(1) & 0.0735(2) & 2     & 0.0416(2) & 0.0182(3) & 0.0023(2) & 0.0465(2) & 2     \\
        & PNS   & 0.0587(4) & 0.0270(4) & 0.0035(4) & 0.0688(4) & 4     & 0.0303(5) & 0.0141(4) & 0.0017(5) & 0.0341(5) & 4.75  \\
        & DNS   & 0.0234(7) & 0.0090(7) & 0.0015(7) & 0.0290(7) & 7     & 0.0219(6) & 0.0072(6) & 0.0012(6) & 0.0246(6) & 6     \\
        & AOBPR & 0.0336(6) & 0.0149(6) & 0.0022(6) & 0.0423(6) & 6     & 0.0309(4) & 0.0131(5) & 0.0018(4) & 0.0352(4) & 4.25  \\
        & SRNS  & 0.0479(5) & 0.0222(5) & 0.0030(5) & 0.0575(5) & 5     & 0.0173(7) & 0.0069(7) & 0.0010(7) & 0.0202(7) & 7  \\
        & MixGCF& 0.0630(2) & 0.0305(1) & 0.0037(3) & 0.0721(3) & 2.25  & 0.0403(3) & 0.0189(2) & 0.0022(3) & 0.0444(3) & 2.75  \\
    \cline{2-12}
        & AutoSample & 0.0636$^*$(1) & 0.0289(2) & 0.0038(1) & 0.0745$^*$(1) & 1.25  & 0.0447$^*$(1) & 0.0203$^*$(1) & 0.0025$^*$(1) & 0.0499$^*$(1) & 1     \\
    \hline
    \multirow{7}{*}{\rotatebox{90}{NGCF}}
        & RNS   & 0.0464(3) & 0.0204(3) & 0.0028(2) & 0.0548(3) & 2.75  & 0.0269(5) & 0.0114(5) & 0.0015(5) & 0.0304(5) & 5 \\
        & PNS   & 0.0369(4) & 0.0160(4) & 0.0021(4) & 0.0422(4) & 4     & 0.0289(4) & 0.0128(4) & 0.0016(4) & 0.0324(4) & 4 \\
        & DNS   & 0.0209(5) & 0.0090(5) & 0.0013(5) & 0.0259(5) & 5     & 0.0220(7) & 0.0070(7) & 0.0012(7) & 0.0247(7) & 7 \\
        & AOBPR & 0.0120(7) & 0.0051(7) & 0.0008(7) & 0.0150(7) & 7     & 0.0222(6) & 0.0091(6) & 0.0013(6) & 0.0258(6) & 6 \\
        & SRNS  & 0.0188(6) & 0.0086(6) & 0.0011(6) & 0.0220(6) & 6     & 0.0316(3) & 0.0132(3) & 0.0018(3) & 0.0349(3) & 3\\
        & MixGCF& 0.0473(2) & 0.0234(2) & 0.0028(2) & 0.0550(2) & 2     & 0.0331(1) & 0.0148(1) & 0.0026(1) & 0.0507(1) & 1 \\
    \cline{2-12}
        & AutoSample & 0.0485$^*$(1) & 0.0246$^*$(1) & 0.0029(1) & 0.0566$^*$(1) & 1     & 0.0329(2) & 0.0137(2) & 0.0019(2) & 0.0367(2) & 2 \\
    \hline
    \end{tabular}
    }
    \begin{tablenotes}
    \footnotesize
    \item [1] $^*$ denotes improvements over baselines are statistically significant with p < 0.05.
    \end{tablenotes}
    \label{tab:main}
\end{table*}

\subsection{Overall Performance(RQ1)}

The overall performance of our AutoSample method and other baselines on three basic models and four benchmark datasets are reported in Table \ref{tab:main}. We summarize our observations below.

First, we can observe that our AutoSample consistently outperforms other negative samplers in the experiments. Considering that the four metrics we employed evaluate various aspects of the results, we assess the performance based on the average rank across the four metrics. Our AutoSample ranks first in 9 out of 12 cases, and even in the remaining three cases, it secures the second rank. These results highlight the superiority of our AutoSample method in terms of model performance.

Furthermore, we can observe that solely targeting hard negative samples does not guarantee constant improvements in all cases. The best-performing baseline methods (excluding our AutoSample) vary from heuristic-sampling methods, like RNS or PNS, to hard negative sampling method, like MixGCF, among all 12 cases. Such an observation validates our hypothesis that the negative sampler needs to be carefully selected given the model and dataset, which motivates the automated negative sampling problem.

Finally, it is surprising to find out that even though AutoSample limits its search space to classic methods, as stated in Section 
\ref{sec:method_search_space}, it can surpass state-of-the-art sampling methods like SRNS or MixGCF. This finding emphasizes the potential benefits of transforming traditional negative sampling, shown in Equation \ref{eq:loss_rec}, into an automated negative sampling problem, shown in Equation \ref{eq:loss_sample}. By introducing more flexibility and adaptability into the negative sampling process, improved performance can be achieved.

\subsection{Ablation on The Amount of Negative Samples(RQ2)}

\begin{table*}[!htbp]
    \centering
    \caption{Ablation over the Number of Negative Samples on Taobao2015 Dataset}
    \label{tab:num_neg}
    \resizebox{.95\textwidth}{!}{
    \begin{tabular}{c|c|cccc|c|cccc|c}
    \hline
        & \multirow{2}{*}{Sampler} & \multicolumn{5}{c|}{T = 2} & \multicolumn{5}{c}{T = 3} \\
    \cline{3-12}
        & & R@20 & N@20 & P@20 & H@20 & Rank & R@20 & N@20 & P@20 & H@20 & Rank \\
    \hline
        \multirow{5}{*}{\rotatebox{90}{MF}} 
        & RNS       & 0.0629(2) & 0.0419(2) & 0.0094(2) & 0.1706(2) & 2 & 0.0647(2) & 0.0437(2) & 0.0097(2) & 0.1753(2) & 2 \\
        & PNS       & 0.0366(5) & 0.0214(5) & 0.0057(5) & 0.1031(5) & 5 & 0.0388(5) & 0.0229(5) & 0.0060(5) & 0.1079(5) & 5 \\
        & DNS       & 0.0554(3) & 0.0405(3) & 0.0081(3) & 0.1510(3) & 3 & 0.0555(3) & 0.0405(3) & 0.0081(3) & 0.1511(3) & 3 \\
        & AOBPR     & 0.0538(4) & 0.0389(4) & 0.0079(4) & 0.1469(4) & 4 & 0.0525(4) & 0.0387(4) & 0.0078(4) & 0.1445(4) & 4 \\
    \cline{2-12}
        & AutoSample & 0.0688$^*$(1) & 0.0449$^*$(1) & 0.0103$^*$(1) & 0.1854$^*$(1) & 1 & 0.0663$^*$(1) & 0.0455$^*$(1) & 0.0099(1) & 0.1793$^*$(1) & 1 \\
    \hline
        \multirow{5}{*}{\rotatebox{90}{LightGCN}} 
        & RNS       & 0.0692(2) & 0.0478(2) & 0.0103(2) & 0.1856(2) & 2     & 0.0695(1) & 0.0479(1) & 0.0104(1) & 0.1867(1) & 1 \\
        & PNS       & 0.0536(5) & 0.0343(5) & 0.0082(4) & 0.1461(5) & 4.75  & 0.0499(5) & 0.0307(5) & 0.0077(5) & 0.1365(5) & 5 \\
        & DNS       & 0.0554(4) & 0.0405(4) & 0.0081(5) & 0.1509(4) & 4.25  & 0.0553(4) & 0.0405(4) & 0.0081(4) & 0.1509(4) & 4 \\
        & AOBPR     & 0.0643(3) & 0.0449(3) & 0.0095(3) & 0.1736(3) & 3     & 0.0661(2) & 0.0456(2) & 0.0098(2) & 0.1780(2) & 2 \\
    \cline{2-12}
        & AutoSample & 0.0725$^*$(1) & 0.0493$^*$(1) & 0.0108$^*$(1) & 0.1932$^*$(1) & 1     & 0.0644(3) & 0.0449(3) & 0.0095(3) & 0.1732(3) & 3 \\
    \hline
        \multirow{5}{*}{\rotatebox{90}{NGCF}} 
        & RNS       & 0.0573(2) & 0.0362(3) & 0.0086(2) & 0.1555(2) & 2.25  & 0.0571(2) & 0.0360(4) & 0.0087(2) & 0.1554(2) & 2.5  \\
        & PNS       & 0.0399(5) & 0.0229(5) & 0.0063(5) & 0.1115(5) & 5     & 0.0427(5) & 0.0248(5) & 0.0067(5) & 0.1188(5) & 5 \\
        & DNS       & 0.0553(3) & 0.0405(2) & 0.0081(4) & 0.1506(3) & 3     & 0.0554(4) & 0.0406(2) & 0.0081(4) & 0.1512(3) & 3.25 \\
        & AOBPR     & 0.0545(4) & 0.0350(4) & 0.0082(3) & 0.1494(4) & 3.75  & 0.0548(3) & 0.0363(3) & 0.0082(3) & 0.1502(4) & 3.75 \\
    \cline{2-12}
        & AutoSample & 0.0706$^*$(1) & 0.0481$^*$(1) & 0.0106$^*$(1) & 0.1904$^*$(1) & 1    & 0.0638$^*$(1) & 0.0437$^*$(1) & 0.0094$^*$(1) & 0.1725$^*$(1) & 1 \\
    \hline
    \end{tabular}
    }
    \begin{tablenotes}
    \footnotesize
    \item [1] $^*$ denotes statistically significant improvements over baselines with p < 0.05.
    \end{tablenotes}
\end{table*}

In Table \ref{tab:main}, AutoSample constantly performs better than other baselines in most cases. But cautious readers may question whether such performance improvement is solely due to the increase in the number of negative instances, as AutoSample introduces $T$ times extra negative instances. This can be observed from Equation \ref{eq:loss_loss}, where AutoSample utilizes different samplers to independently sample negative instances and merge them into the sample loss function. In this section, we aim to address this question by comparing baselines with the same amount of negative instances.

We experiment on the Taobao2015 dataset using three basic models. As stated in Section \ref{sec:method_search_space}, the candidate negative samplers top-$T$ ranked ones in terms of performance. To be specific, for cases where $T = 2$, we adopt $\Pi = [\text{RNS}, \text{DNS}]$ for CF and NGCF model, $\Pi = [\text{RNS}, \text{PNS}]$ for LightGCN model. For cases where $T = 3$, we adopt $\Pi = [\text{RNS}, \text{PNS}, \text{DNS}]$ for all three models. 

As observed in Table \ref{tab:num_neg}, our AutoSample achieves the best performance in 5 out of 6 cases, which demonstrates that the performance improvement of AutoSample is not solely due to the increase in the number of negative instances, but rather the result of a better alignment between the sampler and other components.

It is also interesting to notice that the influence of increasing the number of negative instances varies. For instance, as $T$ increases, the performance of AOBPR decreases for the CF model while it increases for the LightGCN model. This phenomenon indicates that increasing the number of negative samples does not necessarily lead to improved performance.

\subsection{Ablation on the Search Efficiency(RQ3)}

In this section, we conduct an ablation study of the search efficiency of AutoSample. We measure the search efficiency by the total training time, which includes the time for obtaining the optimal search parameter $\alpha^*$ and the corresponding model parameter $W^*$. We compare our method with a grid search that exhaustively runs all candidate negative samplers. The result is shown in Figure \ref{fig:Time}. We can easily observe that our search algorithm can more efficiently explore the search space than an exhaustive grid search.

\begin{figure}[!htbp]
    \centering
    \subfigure[Alibaba]{
    \begin{minipage}[t]{0.2\textwidth}
    \centering
    \includegraphics[width=\textwidth]{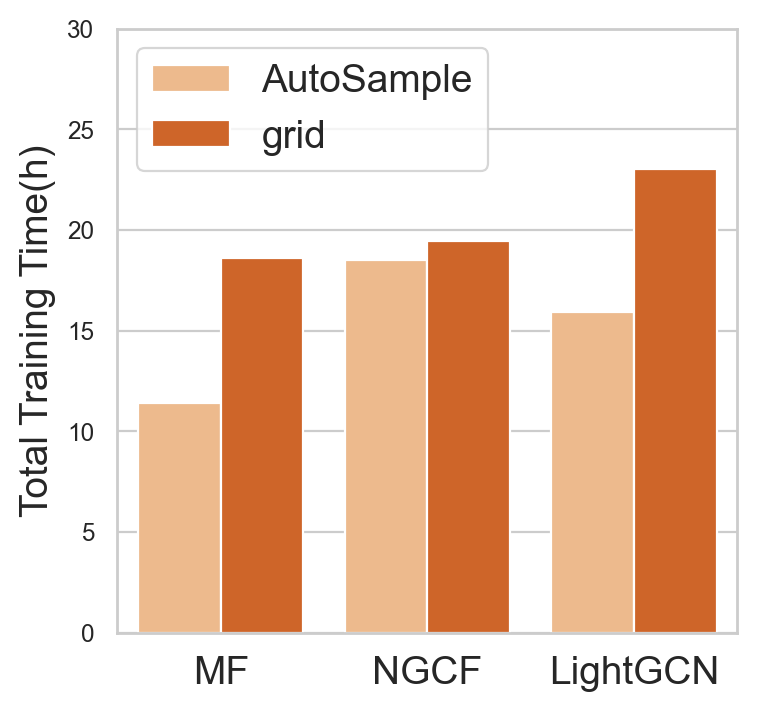}
    \label{fig:Time-ali}
    \end{minipage}
    }
    \subfigure[Amazon]{
    \begin{minipage}[t]{0.2\textwidth}
    \centering
    \includegraphics[width=\textwidth]{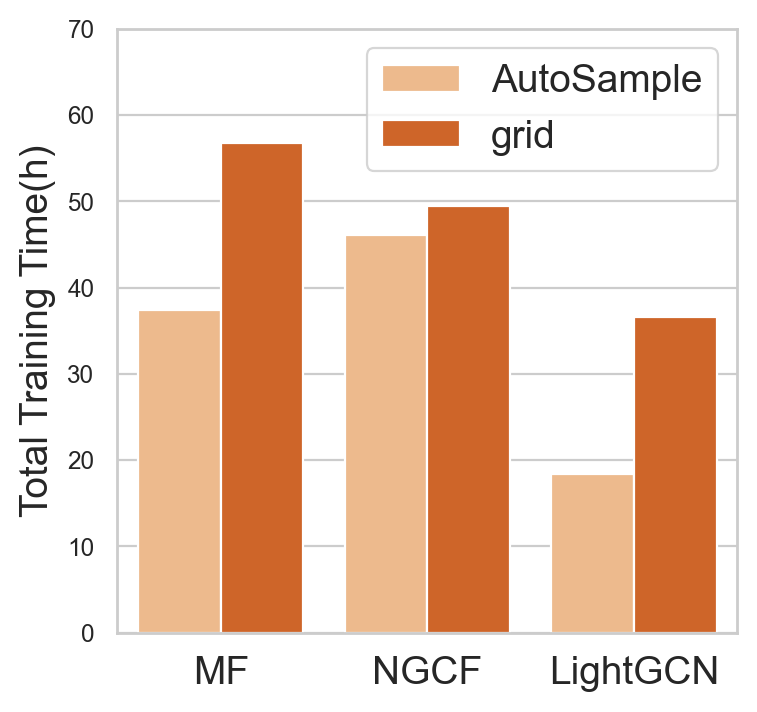}
    \label{fig:Time-amazon}
    \end{minipage}
    }
    \caption{Total Training Time on Alibaba and Amazon.}
    \label{fig:Time}
\end{figure}

\subsection{Ablation on Search Space(RQ4)}
\label{sec:search_space}

In this section, we further investigate the influence of search space over the final result. As we stated in Section \ref{sec:method_search_space}, we select top-$T$ negative samplers based on their performance as candidate samplers for AutoSample. We aim to validate the effectiveness of this selection criteria empirically.

As we can observe from Table \ref{tab:search_space}, the rank of the corresponding AutoSample positively correlates with the average rank of the search space, denoted as \textit{S.S.A.R.}. Here we calculate the \textit{S.S.A.R.} as
$\textit{S.S.A.R.} = \sum_{i=t}^{\text{T}} \text{rank}(\pi_t) / \text{T},$ 
where $\text{rank}(\pi_t)$ refers to the rank of the negative sampler $\pi_t$ and is also reported in Table \ref{tab:main}.
This empirical finding provides further evidence for the feasibility and effectiveness of our selection criteria discussed in the previous section.

\subsection{Ablation on Retraining Scheme(RQ4)}

In this section, we investigate the influence of the retraining scheme mentioned in Section \ref{sec:method_retrain_process}. We conduct experiments on the Taobao2015 and Amazon datasets using all three basic models. We compare our customized retraining scheme (referred to as \textit{custom}) with random initialization (referred to as \textit{random}) in Table \ref{tab:retrain}. As evident from the results, the customized retraining scheme, which utilizes the best-performing parameters from the search process as model initialization, consistently outperforms random initialization. This ablation study establishes the feasibility of our retraining scheme.

\begin{table}[!htbp]
    \centering
    \caption{Ablation over Search Space on Taobao2015 Dataset.}
    \label{tab:search_space}
    \begin{tabular}{c|c|ccc|c}
    \hline
        & \multirow{2}{*}{Metric} & \multicolumn{3}{c|}{T=2} & T=3 \\
    \cline{3-6}
        & & R+P & R+D & D+AO & R+P+D  \\
    \hline
        \multirow{6}{*}{\rotatebox{90}{MF}} 
        & R@20 & 0.0654(3) & 0.0688(1) & 0.0556(4) & 0.0663(2) \\
        & N@20 & 0.0432(3) & 0.0449(2) & 0.0406(4) & 0.0455(1) \\
        & P@20 & 0.0099(2) & 0.0103(1) & 0.0082(4) & 0.0099(2) \\
        & H@20 & 0.1734(3) & 0.1854(1) & 0.1516(4) & 0.1793(2) \\
        & Rank & 2.75 & 1.25 & 4 & 1.75 \\
        & S.S.A.R. & 4.13 & 3.13 & 4.5  & 4.08 \\
    \hline
        \multirow{6}{*}{\rotatebox{90}{LightGCN}} 
        & R@20 & 0.0725(1) & 0.0554(3) & 0.0554(3) & 0.0644(2) \\ 
        & N@20 & 0.0493(1) & 0.0405(3) & 0.0405(3) & 0.0449(2) \\ 
        & P@20 & 0.0108(1) & 0.0081(3) & 0.0081(3) & 0.0095(2) \\ 
        & H@20 & 0.1932(1) & 0.1511(3) & 0.1510(4) & 0.1732(2) \\
        & Rank & 1 & 3 & 3.25 & 2 \\
        & S.S.A.R. & 4.25 & 5.5 & 6.5 & 5.33 \\
    \hline
        \multirow{6}{*}{\rotatebox{90}{NGCF}}
        & R@20 & 0.0655(2) & 0.0706(1) & 0.0550(4) & 0.0638(3) \\
        & N@20 & 0.0437(2) & 0.0481(1) & 0.0403(4) & 0.0437(2) \\
        & P@20 & 0.0098(2) & 0.0106(1) & 0.0081(4) & 0.0094(3) \\
        & H@20 & 0.1768(2) & 0.1904(1) & 0.1503(4) & 0.1725(3) \\
        & Rank & 2 & 1 & 4 & 2.75 \\
        & S.S.A.R. & 4.38 & 3.25 & 4.13 & 4.08 \\
    \hline
    \end{tabular}
    \begin{tablenotes}
    \footnotesize
    \item [1] Here \textit{R}, \textit{P}, \textit{D}, and \textit{AO} stand for RNS, PNS, DNS and AOBPR, respectively. \textit{S.S.A.R.} is an abbreviation for search space average rank.
    \end{tablenotes}
    \label{tab:my_label}
\end{table}

\subsection{Case Study(RQ5)}
In this section, we perform a case study to examine the changes in search parameter $\bm{\alpha}$ during the search process. The evolving trends of $\bm{\alpha}$ at each epoch are reported in Figure \ref{fig:case}. As observed, the evolving patterns of $\bm{\alpha}$ vary across different datasets and models. For cases like NGCF on Alibaba and LightGCN on Amazon, AutoSample exhibits a curriculum learning behaviour, gradually adjusting the negative instances to facilitate better learning. In contrast, for other cases, AutoSample determines the optimal sampler at an early stage, resembling a neural architecture search approach, and trains the model given the searched sampler in the later epochs.

\begin{table}[!htbp]
    \centering
    \caption{Ablation over Retraining Scheme on Taobao2015 and Amazon Datasets.}
    \label{tab:retrain}
    \begin{tabular}{c|c|cc|cc}
    \hline
    & \multirow{2}{*}{Metric} & \multicolumn{2}{c|}{Taobao2015} & \multicolumn{2}{c}{Amazon} \\
    \cline{3-6}
    & & \textit{random} & \textit{custom} & \textit{random} & \textit{custom} \\
    \hline
    \multirow{4}{*}{\rotatebox{90}{MF}} 
    & R@20 & 0.0556 & 0.0688 & 0.0281 & 0.0354 \\
    & N@20 & 0.0406 & 0.0449 & 0.0123 & 0.0155 \\
    & P@20 & 0.0082 & 0.0103 & 0.0016 & 0.0020 \\
    & H@20 & 0.1516 & 0.1854 & 0.0323 & 0.0391 \\
    \hline
    \multirow{4}{*}{\rotatebox{90}{LightGCN}} 
    & R@20 & 0.0696 & 0.0725 & 0.0416 & 0.0447 \\
    & N@20 & 0.0479 & 0.0493 & 0.0182 & 0.0203 \\
    & P@20 & 0.0105 & 0.0108 & 0.0023 & 0.0025 \\
    & H@20 & 0.1873 & 0.1932 & 0.0465 & 0.0499 \\
    \hline
    \multirow{4}{*}{\rotatebox{90}{NGCF}} 
    & R@20 & 0.0553 & 0.0706 & 0.0269 & 0.0329 \\
    & N@20 & 0.0405 & 0.0481 & 0.0114 & 0.0137 \\
    & P@20 & 0.0081 & 0.0106 & 0.0015 & 0.0019 \\
    & H@20 & 0.1509 & 0.1904 & 0.0304 & 0.0367 \\
    \hline
    \end{tabular}
\end{table}

\begin{figure}[!htbp]
    \centering
    \begin{minipage}[t]{0.2\textwidth}
    \centering
    \includegraphics[width=\textwidth]{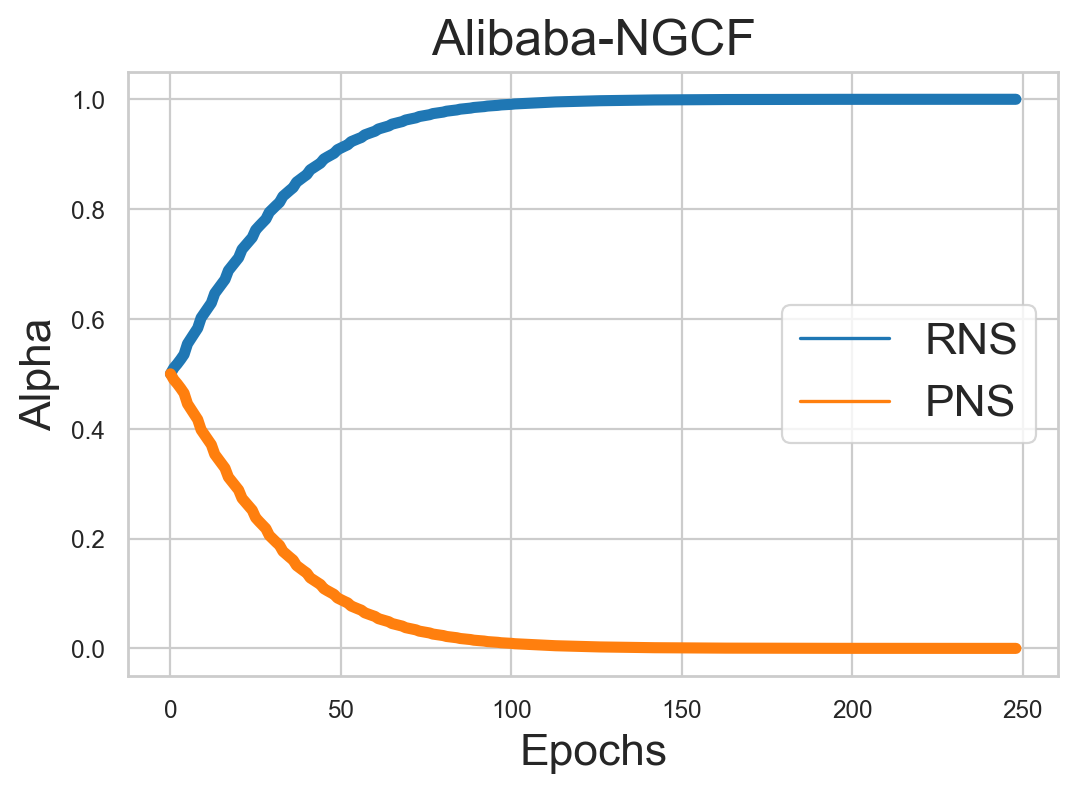}
    \label{fig:case-ali-ngcf}
    \end{minipage}
    \begin{minipage}[t]{0.2\textwidth}
    \centering
    \includegraphics[width=\textwidth]{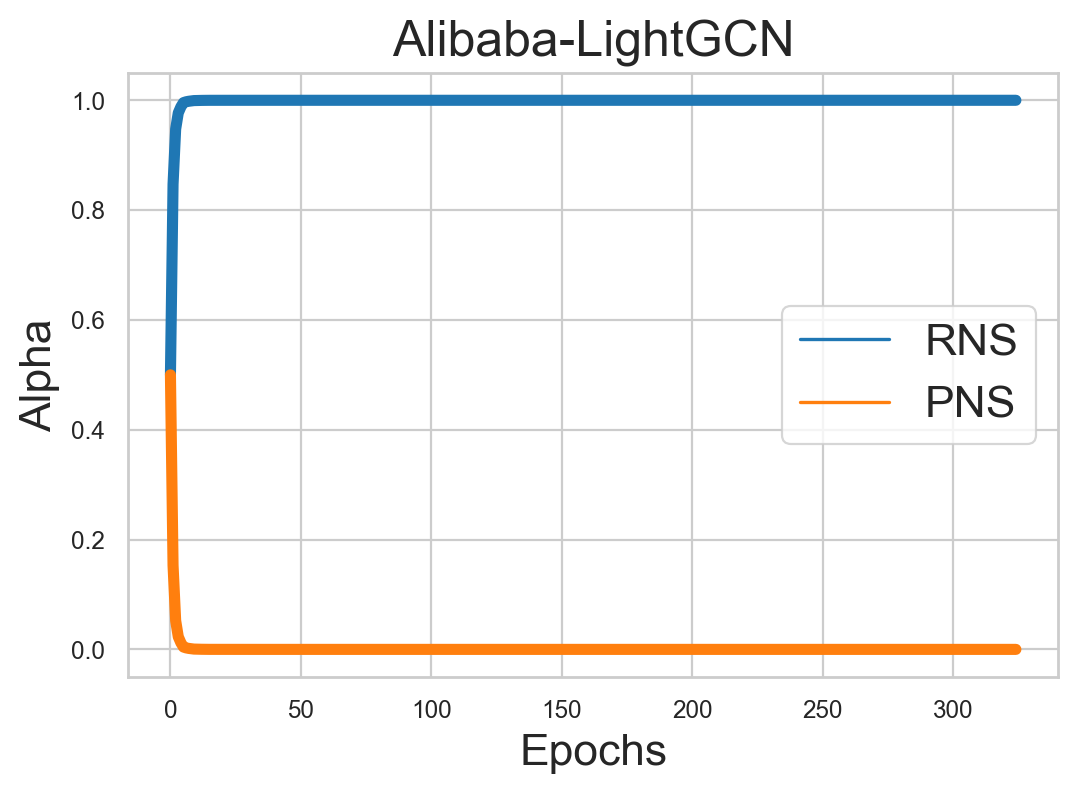}
    \label{fig:case-ali-lightgcn}
    \end{minipage}
    \begin{minipage}[t]{0.2\textwidth}
    \centering
    \includegraphics[width=\textwidth]{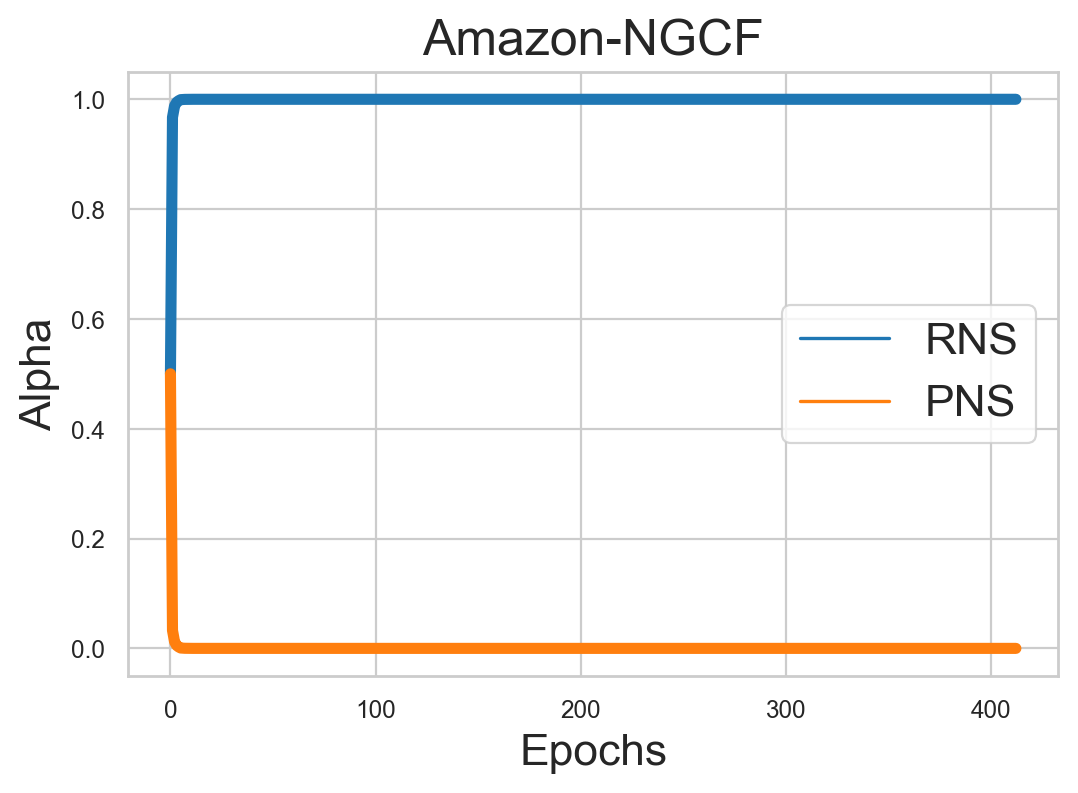}
    \label{fig:case-amazon-ngcf}
    \end{minipage}
    \begin{minipage}[t]{0.2\textwidth}
    \centering
    \includegraphics[width=\textwidth]{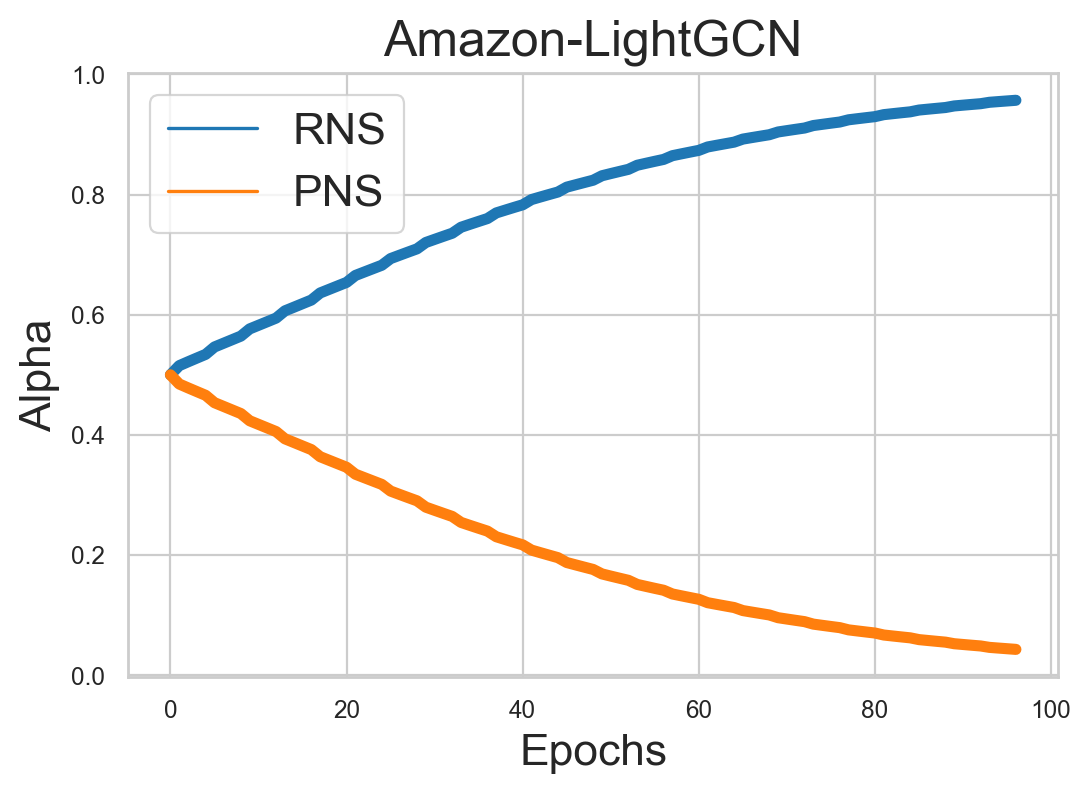}
    \label{fig:case-amazon-lightgcn}
    \end{minipage}
    \caption{Case Study of $\alpha$ on Alibaba and Amazon datasets.}
    \label{fig:case}
\end{figure}

\section{Related Work}

We discuss how our work is situated in two research topics: (i) negative sampling and (ii) automated machine learning and its application in recommender systems. 

\subsection{Negative Sampling}

Negative sampling has been an active topic for both industry and academia. It can mainly be summarized into four classes~\cite{NS-Survey}. 
Heuristic sampling methods rely on specific heuristics to guide the sampling process, including item popularity~\cite{PNS}, uniform sampling~\cite{BPR} or employing Markov chains~\cite{MCNS,NNCF}. 
Hard negative sampling methods, including DNS~\cite{DNS}, AOBPR~\cite{AOBPR}, SRNS~\cite{SRNS}, NSCaching~\cite{NS-Caching}, and ReinforceNS~\cite{ReinforceNS}, assign higher probabilities to instances with significant prediction score, aiming to enhance model training. 
Adversarial Sampling methods, including IRGAN~\cite{IRGAN}, AdvIR~\cite{AdvIR}, and GANRec~\cite{GANRec}, treat the sampling of negative samples as a generative problem and utilize GAN-based methods to improve model training. 
Graph-based sampling methods, such as RWS~\cite{RWS}, RecNS~\cite{RecNS}, and MixGCF~\cite{MixGCF} tend to guide the sampling process depending on the underlying graph topology, typically applied in graph-based models.
Additionally, some approaches treat all unobserved data as negative~\cite{ENMF,Gramian}. However, this aspect is beyond the scope of discussion in this work.

Our work builds upon prior studies and can be combined with any negative sampling method. It intends to choose the optimal negative sampling methods depending on models and datasets.

\subsection{Automated Machine Learning and its Application in Recommender System}

Automated machine learning has emerged as a promising approach to automatically select specific tasks or datasets, eliminating the necessity for human experts~\cite{NAS,NAO,DARTS,Large-Scale-Evolution}. It has substantially advanced various domains, such as computer vision~\cite{NAO,DARTS}, natural language processing~\cite{nas-nlp} and feature representation learning~\cite{FRL-Survey}. There are two key aspects to automated machine learning: \textit{selection space} and \textit{selection algorithm}. The \textit{selection space} comprises all potential choices and is usually task-dependent. The \textit{selection algorithm} aims to explore the selection space efficiently and can generally be categorized into three classes: controller-based~\cite{NAS}, evolution-based~\cite{Large-Scale-Evolution} and gradient-based~\cite{NAO,DARTS}.

Automated machine learning has found widespread application in recommender systems~\cite{DRS-Survey}. It is widely adopted for tasks such as selecting suitable embedding dimensions~\cite{NIS}, identifying informative features~\cite{OptFS}, determining beneficial feature interactions~\cite{AutoFIS}, choosing integration function~\cite{AutoFeature,OptInter} or designing comprehensive architectures~\cite{AutoPI}. 

Our work distinguishes itself from the existing research by addressing the challenge of negative sampling in recommender systems. It represents a distinct direction orthogonal to previous studies, and there is potential for further advancements by combining our approach with existing methodologies.

\section{Conclusion}
\label{sec:conclusion}

In this paper, we first challenge the previous negative sampling schemes which explicitly target the hard negative instances and develop negative samplers orthogonal to the implicit datasets and recommendation models. Inspired by the intuition of neural architecture search, we propose a hypothesis that the negative sampler, which generates the negative instances, should be matched with the positive implicit data and the model. Such a hypothesis is empirically proven as the best-performing negative samples vary as the dataset and model change. To solve such a problem, we formulate the automated negative sampling problem and conduct \textit{instance-to-loss} approximation to make it end-to-end trainable. A method called AutoSample is also introduced to efficiently and effectively solve the problem. Extensive experiment over four benchmark datasets and three recommendation models proves the feasibility of AutoSample. Several ablation studies also investigate different settings in AutoSample. Moreover, we also visualize the search result in AutoSample, suggesting that the AutoSample somehow correlates with curriculum learning in obtaining the model parameters.

\normalem
\bibliographystyle{ACM-Reference-Format}
\bibliography{main.bib}
\appendix
\section{Ethical Considerations}
The proposed AutoSample focuses on sampling 
method in recommendation models. As with other recommendation models, our method may cause filter bubbles or echo chamber phenomena to users of the service. Therefore, when deploying our AutoSample with the recommendation model, it is necessary to combine some existing debiasing techniques to reduce these adverse effects.

\section{Experimental Settings}

\subsection{Dataset Description}
\label{ref:appendix_exp_dataset}

In this paper, we evaluate our method on the following benchmark datasets: Taobao2014, Taobao2015, Alibaba and Amazon. The statistics of these datasets are summarized in Table \ref{table:dataset}. We randomly split the datasets into training, validation and testing sets by 3:1:1. Belows list the detailed description of each dataset:
\begin{itemize}[topsep=0pt,noitemsep,nolistsep,leftmargin=*]
    \item Taobao2014\footnote{https://tianchi.aliyun.com/dataset/46} contains real user behaviour data from 2014 provided by Alibaba Group. We filtered out user and item nodes with fewer than 10 interactions.
    \item Taobao2015\footnote{https://tianchi.aliyun.com/dataset/53} is a real-world dataset accumulated on Tmall / Taobao and the app Alipay in 2015. We filtered out user and item nodes with fewer than 10 interactions.
    \item Alibaba dataset~\cite{MCNS} is constructed based on the data from another large E-commerce platform, which contains the user’s purchase records and items’ attribute information. The data are organized as a user-item graph for the recommendation.
    \item Amazon\footnote{http://jmcauley.ucsd.edu/data/amazon/links.html} is a large E-commerce dataset introduced in previous work~\cite{MCNS,MixGCF}, which contains purchase records and review texts whose time stamps span from May 1996 to July 2014. In the experiments, we take the data from the commonly-used “electronics” category to establish a user-item graph.
\end{itemize}

\begin{table}[htbp]
    \caption{Dataset Statistics}
    \centering
    \begin{tabular}{c|c|c|c|c}
    \hline
        Datasets & \#user & \#item & \#interactions & Density(\textperthousand) \\
    \hline
        Taobao2014  & 8844      & 39103     & 749438        & 2.1671 \\
        Taobao2015  & 92605     & 9842      & 1332602       & 1.4621 \\
        Alibaba     & 106042    & 53591     & 907470        & 0.1597 \\
        Amazon      & 192403    & 63001     & 1689188       & 0.1394 \\
    \hline
    \end{tabular}
    \label{table:dataset}
\end{table}

\subsection{Recommendation Models} 
\label{ref:appendix_exp_rec}
To verify the effectiveness of our method, we perform experiments on the following three recommendation models.

\begin{itemize}[topsep=0pt,noitemsep,nolistsep,leftmargin=*]
    \item Matrix Factorization~\cite{BPR} is a commonly-used model in implicit recommender systems to learn both the users' and items' embedding. It calculates the relevance score between a user and an item by calculating the inner product of their embeddings.
    \item LightGCN~\cite{LightGCN} is a lightweight version of the GCN~\cite{GCN}, customized for recommender system. It only contains the neighbourhood aggregation component for collaborative filtering.
    \item NGCF~\cite{NGCF} is a widely-used graph recommendation model. It proposes to integrate the user-item interactions into the embedding process by injecting the collaborative signal into the embedding process in an explicit manner.
\end{itemize}

\subsection{Baselines}
\label{ref:appendix_exp_baseline}
We compare our proposed method with various baseline methods as follows. For all sampling-based methods, we sample one negative instance for each positive instance following previous works~\cite{IRGAN}. 
\begin{itemize}[topsep=0pt,noitemsep,nolistsep,leftmargin=*]
    \item RNS~\cite{BPR}, or random sampling, selects negative samples from all candidates with a uniform probability.
    \item PNS~\cite{PNS}, or popularity-based sampling, selects negative samples based on the popularity of each item. We set the hyper-parameter $\beta=0.75$ following previous work~\cite{PNS}.
    \item AOBPR~\cite{AOBPR} deliberately selects the hard negative samples with the relevance scores in the recommendation models. It calculates the relevance scores for all unobserved instances and chooses the one with the highest relevance score as the negative instance.
    \item DNS~\cite{DNS}, or dynamic negative sampling, also aims to select the hard negative samples with the relevance scores in the recommendation models. It selects a fixed number of unobserved instances as candidates and chooses the one with the highest relevance score as the negative instance.
    \item SRNS~\cite{SRNS} is a SOTA sampling method that wishes to select hard negative samples. It distinguished the hard negative and false negative instances by the variance of relevance scores and favoured the false negative ones with high variance. And it keeps a designed memory that only stores a few important candidates, which boosts the sampling efficiency.
    \item MixGCF~\cite{MixGCF} is a recent method targeting graph-based negative sampling. Instead of sampling negative instances, it synthesizes them by aggregating the neighbours' embeddings. Notes that for MixGCF, we only compare it on LightGCN and NGCF as it is a method targeting graph-based recommendation models instead of general ones.
\end{itemize}


\subsection{Platform}
All experiments are conducted on a Linux server with 8 Intel Xeon Gold 6154 cores, 64 GB memory and one Nvidia-V100 GPU with PCIe connections.

\subsection{Reproducibility and Hyper-parameters}
\label{ref:appendix_exp_repro}
Our implementation\footnote{https://anonymous.4open.science/r/AutoSample} is adapted from the official PyTorch implementation for the MixGCF baseline\footnote{https://github.com/huangtinglin/MixGCF}. For other baseline methods, we reuse the official implementation for MixGCF and SRNS\footnote{https://github.com/dingjingtao/SRNS/}~\cite{SRNS}. 
We re-implement the RNS~\cite{BPR}, PNS~\cite{PNS}, AOBPR~\cite{AOBPR} and DNS~\cite{DNS} in our implementation due to either lack of official implementations or the official repositories are supported by early-stage deep learning frameworks. 

We also provide the hyper-parameter tuning details for our experiments. For AutoSample, (i) General hyper-parameters: We set the embedding dimension to 64 and batch size to 1024. We set the additional network hyper-parameter introduced by LightGCN and NGCF the same as the previous paper~\cite{MixGCF}. Following previous work~\cite{MixGCF}, Adam optimizer is adopted.
We tune the optimal learning ratio from \{3e-3, 1e-3, 3e-4, 1e-4\} and $l_2$ regularization from \{1e-2, 1e-3, 1e-4, 0\}. (ii) Additional hyper-parameters: For the selection parameter $\alpha$, we tune its learning ratio from \{3e-2, 1e-2, 3e-3, 1e-3, 3e-4, 1e-4, 3e-5, 1e-5\}  and exclude any regularization as we do not want to incorporate prior information. During the re-training process, we reuse the optimal learning ratio and $l_2$ regularization. For MixGCF, we select its optimal hyperparameter reported in the paper. For SRNS, as it was never evaluated on the four benchmarks in our settings, we select the optimal hyperparameter from the same hyperparameter domain of AutoSample. We will report the exact optimal hyperparameters used in the experiments after the paper is accepted.

\end{document}